\documentclass[12pt,a4paper]{amsart}
\usepackage{amsmath,amssymb, amsbsy}
\usepackage{color,psfrag}
\usepackage{graphicx, caption}
\usepackage{enumerate}
\usepackage{stackrel}
\usepackage[bookmarks=false]{hyperref}

\usepackage{mathtools}
\mathtoolsset{showonlyrefs}

\renewcommand{\P}{\mathbb{P}}
\newcommand{\Q}{\mathbb{Q}}

\newcommand{\ieq}{\begin{equation}}
\newcommand{\eeq}{\end{equation}}
\newcommand{\ieqa}{\begin{eqnarray}}
\newcommand{\eeqa}{\end{eqnarray}}
\newcommand{\ieqas}{\begin{eqnarray*}}
\newcommand{\eeqas}{\end{eqnarray*}}

\theoremstyle{plain}
\newtheorem{theorem}{Theorem} [section]

\def\neweq#1{\begin{equation}\label{#1}}
\def\endeq{\end{equation}}

\theoremstyle{definition}

\numberwithin{figure}{section}

\def\F{{\mathcal F}}

\def\F{{\mathcal F}}

\newcommand{\average}{{\mathchoice {\kern1ex\vcenter{\hrule
height.4pt width 6pt depth0pt} \kern-11pt} {\kern1ex\vcenter{\hrule height.4pt width 4.3pt depth0pt} \kern-7pt} {} {} }}

\newcommand{\ofp}{(\Omega, \mathcal{F}, \P)}

\def\ind{{\mathchoice{1\mskip-4mu\mathrm l}{1\mskip-4mu\mathrm l}
{1\mskip-4.5mu\mathrm l}{1\mskip-5mu\mathrm l}}}

\usepackage{graphicx}
\graphicspath{
{ImagesCEC/}
}

\title{Intertemporal Cost-efficient Consumption}

\begin{document}

\author[M. Elizalde, S. Sturm]{Mauricio Elizalde, Stephan Sturm}
\address{}
\email{}

\maketitle


\vskip2mm
\noindent
{\footnotesize

Mauricio Elizalde,
Departamento de Matem\'aticas,
Universidad Aut\'onoma de Madrid. Calle Francisco Tom\'as y Valiente, 7, 28049, Madrid, Spain. 
(e-mail mauricio.elizalde@estudiante.uam.es)

Stephan Sturm,
Department of Mathematical Sciences,
Worcester Polytechnic Institute, 100 Institute Road, Worcester,
MA 06109, USA.
(e-mail: ssturm@wpi.edu)

}
 
\vskip3mm\noindent

\begin{abstract}
We aim to provide an intertemporal, cost-efficient consumption model that extends the consumption optimization inspired by the Distribution Builder, a tool developed by Sharpe, Johnson, and Goldstein. The Distribution Builder enables the recovery of investors' risk preferences by allowing them to select a desired distribution of terminal wealth within their budget constraints. 

This approach differs from the classical portfolio optimization, which considers the agent's risk aversion modeled by utility functions that are challenging to measure in practice. Our intertemporal model captures the dependent structure between consumption periods using copulas. This strategy is demonstrated using both the Black-Scholes and CEV models.

\vspace{5mm}

\end{abstract}

\begin{flushleft}

\end{flushleft}

\section{Introduction}

In this work, we address the problem of intertemporal optimal consumption, a model to consume an investment aiming to attain desired distributions while minimizing costs. We develop our model in the same spirit of the Distribution Builder, a tool developed by Sharpe \textit{et al.} (refer to \cite{goldstein2008choosing} and \cite{sharpe2006asset}), which is designed  to recover the investor's risk preferences by allowing them to arrange blocks in a screen that represent units of probability to shape their own consumption distribution attained to an initial capital.

Our methodology stands in contrast to classical portfolio optimization, which relies on modeling the risk aversion of the agent through utility functions. However, utility functions can be challenging to measure in practice and under certain circumstances they violate most assumptions (\cite{tversky1975critique}). As a result, our intertemporal approach provides an alternative that avoids the complexities associated with utility functions and offers a more flexible and practical solution for achieving optimal consumption distributions.

In their works \cite{goldstein2008choosing} and \cite{sharpe2006asset}, the authors focused on a single-period model and utilized the Distribution Builder with a finite equiprobable probability space $(\Omega,\mathcal{F},\mathbb{P})$, where $\Omega:=\{\omega_i\}^n_{i=1}$, and $\mathbb{P}(\{\omega_i\})=1/n,i=1,...,n$. 
In this context, the wealth distribution $F$ can be represented as follows
$$F(x)=\dfrac{1}{n}\big\vert\{x_i\vert x_i\leqslant x,i=1,...,n\}\big\vert,$$
where $x$ represents the outcome of a random variable $X$ in the $n$ states respectively.
We associate each $x_i$ of the random variable $X$ with one of the $n$ states of the market. 

The expression of the minimum-cost problem in the equiprobable space is given by the following formula:
$$\min_X \sum^n_{i=1}\xi(\omega_i)X(\omega_i),$$
where $\xi$ is the state price process or price kernel. Consequently, each claim $X$ is priced as $\mathbb{E} [\xi X]= \tilde{\mathbb{E}}[X]$. 

Dybvig (\cite{D88a}, \cite{D88b}) showed that the anticomonotonic property is a sufficient condition for cost-efficiency. Two random variables $X$ and $Y$ on a probability space $(\Omega,\mathcal{F},\mathbb{P})$ are considered comonotone if there exists $\Omega_0$ with $\mathbb{P}(\Omega_0)=1$ such that $(X(\omega)-X(\omega'))(Y(\omega)-Y(\omega'))\geqslant 0$ for any $\omega,\omega'\in\Omega_0$. $X$ and $Y$ are anticomonotone if $X$ and $-Y$ are comonotone. Thus, in the pursuit of finding the random variable that achieves the optimal consumption, we aim for it to be anticomonotonic with the price kernel. This involves assigning the largest consumption outcome to the state with the cheapest price, the second largest consumption outcome to the state with the second cheapest price, and so on.

We preserve the spirit of the distribution builder but in a multiperiod model.
A naive approach to propose an extension of the cost-efficient consumption in an intertemporary way for $N$ periods is letting the investor specify a desired distribution of every consumption period $k$ and iterate the terminal market approach, arranging every distribution $X_k$ to be antimonotonic with the price kernel $\xi_k$:

$$
\min_{X_k} \sum^n_{i=1}\xi_k(\omega_i)X_k(\omega_i), \ 1 \leq k \leq N.
$$
But this is economically not sensible as this approach does not capture comonotonicity between consumption flows. 

We introduce a method that is both sophisticated and uncomplicated for introducing a connection between variables using copulas. Copulas are mathematical functions that describe the relationships between multiple variables, featuring uniform distributions on the range [0, 1]. They are particularly useful for indicating interdependencies between random variables. This approach streamlines the agent's choices by breaking them down into two main components: the selection of individual distributions and the adoption of a copula to express their interrelation.

The exploration of cost-efficient consumption remains a current and evolving subject, with recent research including \cite{JZ08}, where the authors described the optimal portfolio choice of an investor optimizing a Cumulative Prospect Theory objective function in a complete discrete market with a equiprobable probability space,  \cite{BBV14}, a work that provides an explicit representation of the lowest cost strategy to achieve a given payoff distribution in the Black-Scholes market, and \cite{BS2022}, where the authors extend the problem to incomplete markets. They found that the main results from the theory of complete markets still hold in adapted form, and reveals that the optimal portfolio selection for preferences characterized by law-invariance and non-decreasing tendencies toward diversification must be perfect cost-efficient.

This, paper is organized as follows. 
In Section \ref{existence}, we establish the existence of cost-efficiency. Specifically, we prove the existence of the portfolio that attains the minimum cost. This is achieved by considering the given consumption distributions that represent the preferences of the agent.
Section \ref{dependency} outlines the manner in which we establish a dependency structure between consumption flows, focusing on the application of the Clayton copula.
In Section \ref{algorithm}, we introduce an algorithm designed to determine the optimal consumption random vector. This algorithm incorporates a dependency structure through the utilization of copulas. It consists of three essential steps: simulating the state-price process, generating values for the consumption distribution, and ultimately computing the solution. We would like to express our gratitude to Carole Bernard of Grenoble Ecole de Management for her ideas to develop this algorithm.
Section \ref{examples} provides concrete examples that demonstrate the application of the strategy under both the Black-Scholes model and the Constant Elasticity of Variance (CEV) model. We employ Clayton copulas in both cases. In the context of the CEV model, computing the state price process is not a straightforward task. As a solution, we derive its moment generation function using the Laplace transform and subsequently estimate it through simulations of the stock price. This approach capitalizes on the equivalence of the CEV model to a square root diffusion.

\section{Existence of the Cost-efficient Consumption}
\label{existence}

In pursuit of the objective to determine a cost-efficient consumption aligned with the agent's preferences, we consider the sum of random variables that represent the consumption at each time instance. As previously mentioned, the distribution of this sum has to be anticomonotonic with the distribution of the state price process. In this section, we provide a proof to establish that the set of random variables which meets the agent preferences is not empty.
Particularly, the following theorem proves that given the consumption distributions, it is possible to find a random vector that satisfies these distributions.

\begin{theorem}\label{thm:exist}
	Let $(X_1, \, X_2, \ldots, \, X_N)$ be a random vector on the atomless standard probability space $\ofp$ with joint cumulative distribution function $F$. The marginal distributions of the vector and $F$ represent the agent preferences. We define $Z = \sum_{k=1}^N X_k$. Let $\tilde{Z}$ be a random variable defined on $\ofp$ that has the same distribution as $Z$. Then, there exists a random vector $(\tilde{X}_1, \, \tilde{X}_2, \ldots, \, \tilde{X}_N)$, which represents the agent consumption at time $1,2, \dots, N$ respectively, with joint cumulative distribution function $F$, such that $\sum_{k=1}^N \tilde{X}_k = \tilde{Z}$ and $\tilde{X}_k$ is distributed as $X_k$ for $k=1,...,N$. It is in general not unique.
\end{theorem}

\begin{proof}
	For the case $N=2$, consider the joint distribution function of the given random variables including the sum, i.e.,
	\[
	H(z,x_1, x_2) = \P\bigl[X_1 + X_2 \leq z, X_1 \leq x_1, X_2, \leq x_2\bigr].
	\]
	Then the goal is to construct a vector $(\tilde{X}_0, \tilde{X}_1, \tilde{X}_2)$ with distribution $H$ which satisfies $\tilde{X}_0 = \tilde{Z}$. This can be done by combining a one-dimensional distribution transform with a multidimensional quantile transform (see \cite{R_DT}). Specifically, define $U_{\tilde{Z}}$ as the distributional transform of $\tilde{Z}$, i.e., by setting
\begin{align*}
h(z, \lambda) &:= \P\bigl[\tilde{Z} < z\bigr] + \lambda \P\bigl[\tilde{Z} = z\bigr]\\
U_{\tilde{Z}}& := h\bigl(\tilde{Z}, U\bigr)
\end{align*}
for some $U \sim \mathcal{U}\bigl((0,1)\bigr)$ independent of $\tilde{Z}$. Then $U_{\tilde{Z}}$ is a standard uniform random variable, and for $G^{-1} $ the generalized inverse (quantile function) of $G$ we have $G^{-1}\bigl(U_{\tilde{Z}}\bigr) = \tilde{Z}$. Note that the existence of such an independent random variable follows from the atomlessness of the probability space (cf. \cite[Proposition A.27]{FS}).

Now consider the conditional distribution of $X_1$ given the sum $X_1 + X_2 = Z$. We have
\[
F_{X_1 \vert Z}(x_1 \vert z) = \P\bigl[X_1 \leq x \, \, \vert \, X_1 + X_2 = Z\bigr].
\]
%
We note that the conditional regular probabilities exist as our probability space is standard.
Therefore
we can define the vector $(\tilde{X}_0, \tilde{X}_1, \tilde{X}_2)$ by generating a standard uniform $V_1$ independent of $U$ and $\tilde{Z}$ (and thus independent of $U_{\tilde{Z}}$) and setting
\begin{align*}
\tilde{X}_0 &:= G^{-1}\bigl(U_{\tilde{Z}}\bigr) = \tilde{Z},\\
\tilde{X}_1 &:= F_{X_1 \vert Z}^{-1}(V \vert \tilde{Z}),\\
\tilde{X}_2 &:= \tilde{Z} - \tilde{X}_1,
\end{align*}
since the third quantile function is trivial as $F_{X_2 \vert Z, X_1}(x_2 \vert Z, X_1) = \ind_{\{Z-X_1 \leq x_2\}}$.
Thus we have to provide a construction of a random vector $(\tilde{X}_0, \tilde{X}_1, \tilde{X}_2)$ that has indeed the distribution $H$ by construction as well as $\tilde{X}_0 = \tilde{Z}$.

We finish by noting that the construction is unique if and only if is independent of the choice of the uniforms $U$ and $V$. This is the case only if the random variable $Z$ has no atom (i.e., the c.d.f. $G$ is continuous) and $X_1$ and $X_2$ are monotone functions of the same random variable, i.e., either comonotonic or anticomonotonic.

For the case $N>2$, we retake the previous procedure renaming $X_1$ as $X_N$ and $X_2$ as $X_1+...+X_{N-1}$. Then we have 
\begin{align*}
\tilde{X}_0 &:= G^{-1}\bigl(U_{\tilde{Z}}\bigr) = \tilde{Z}\\
\tilde{X}_1 &:= F_{X_N \vert Z}^{-1}(V_1 \vert \tilde{Z})\\
\tilde{X}_2^{'} &:= \tilde{Z} - \tilde{X}_1
\end{align*} 
where $\tilde{X}_2^{'}$ has the same distribution as $X_1+...+X_{N-1}$.

We apply the case of $\ N=2\ \  (X_1, \ X_2, \ X_1+X_2 = Z\ )$ for the triple $(X_{N-1}, \ X_1+...+X_{N-2}, \ X_1+...+X_{N-2}+X_{N-1} = Z-X_N)$ and results
\begin{align*}
\tilde{X}_2^{'} &:= \tilde{Z} - \tilde{X}_1\\
\tilde{X}_2 &:= F_{X_{N-1} \vert Z-X_n}^{-1}(V_2 \vert \tilde{Z} - \tilde{X}_1)\\
\tilde{X}_3^{'} &:= \tilde{Z} - \tilde{X}_1 - \tilde{X}_2
\end{align*} 

We apply the procedure $N-1$ times to have
\begin{align*}
\tilde{X}_0 &:= \tilde{Z}\\
\tilde{X}_1 &:= F_{X_{N} \vert Z}^{-1}(V_1 \vert \tilde{Z})\\
\tilde{X}_2 &:= F_{X_{N-1} \vert Z-X_n}^{-1}(V_2 \vert \tilde{Z} - \tilde{X}_1)\\
\vdots\\
\tilde{X}_{N-1} &:= F_{X_2 \vert Z-X_N-...-X_3}^{-1}(V_{N-1} \vert \tilde{Z} - \tilde{X}_1-...-\tilde{X}_{N-2})\\
\tilde{X}_N &:= \tilde{Z} - \tilde{X}_1 -...- \tilde{X}_{N-1}
\end{align*} 
where $V_k$ is a standard uniform random variable independent of $\tilde{Z} - \tilde{X}_1 -...- \tilde{X}_{i-1}$, $1 \leq k \leq N$.
\end{proof}

Once we have proven the existence of such a random vector, we are ready to present how we model the dependence structure and the algorithm to find an optimal random vector that leads us to a cost-efficient strategy.

\section{Dependency structure}
\label{dependency}

We propose a sophisticated and simple way to add a dependency structure of consumption flow via copulas, which are multivariate cumulative distribution functions with uniform marginal distributions on the interval [0,1] and are practical to describe the inter-correlation between random variables. 
In this work, we choose the Clayton copula due to its easiness to model it with one parameter. Clayton copulas are Archimedean copulas useful to model dependent structure as they have only one parameter $\alpha \in [-1,\infty) \setminus \{ 0 \} $ in its expression
$$
C^{\alpha}_{X_1,...,X_N}(x_1,...,x_N) = \left( \left( 1-N + \sum^N_{k=1} F(x_k)^{-\alpha} \right)_+\right)^{-\frac{1}{\alpha}},
$$
which we denote just by $C^{\alpha}$ for short when there is no confusion. 

The interpretation of $\alpha$ is easy and related to the correlation between the random variables: if $\alpha$ is closer to $-1$, we have an anticorrelated structure; if it is close to zero, we have an independent structure; and while it grows to $\infty$, we have a highly correlated structure, as we show in figure~\ref{FigClayton} with for random variables.

\begin{figure}[h]	
\begin{center}
\includegraphics[
trim = 0cm 0cm 0cm 0cm, clip, 
width=14cm]{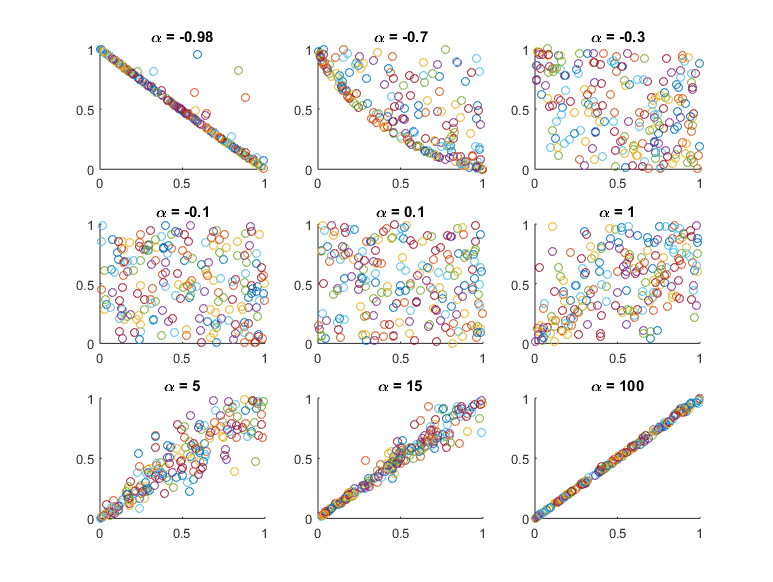}
\caption{Bivariate Clayton Copulas for different values of $\alpha$.}
\label{FigClayton}
\end{center}
\end{figure}

The generator function $\varphi: (0,1] \rightarrow [0,\infty)$ of a Clayton copula with parameter $\alpha$ is  
\begin{equation}
\varphi (u) = \left( u^{-\alpha} - 1 \right) /\alpha.
\label{Clayton1}
\end{equation}
Its inverse function is
\begin{equation}
\varphi^{-1} (u) = \left( 1 + \alpha u \right) ^{-1/\alpha},
\label{Clayton2}
\end{equation}
and its k-th derivative is
\begin{equation}
\varphi^{-1(k)}(u) = (-1)^k (1 + \alpha u)^{-(1 + k\alpha)/\alpha} \prod^{k-1}_{j=0}(1 + j\alpha).
\label{Clayton3}
\end{equation}


We use expressions \eqref{Clayton1}, \eqref{Clayton2} and \eqref{Clayton3} to generate a $N-$tuple vector $(X_1,...,X_N)$ from a Clayton copula with parameter $\alpha$ through the following procedure described in \cite{wu2007simulating}.

\begin{enumerate}
\item[1.] We generate $N$ independent $U(0,1)$ random values, and denote them by $w_1,...,w_N$.
\item[2.] We set $s_k = w_k^{1/k}$ for $k=1,2,...,N-1$.
\item[3.] We compute $v = F_{C^\alpha}^{-1}(w_N)$, where 
$$
F_{C^\alpha}(v) = \dfrac{1}{(N-1)!} \times \int_0^v \varphi^{-1(N)} (w) [\varphi(w)]^{N-1} \varphi' (w) dw,
$$
and can be expressed as (see \cite{genest2001multivariate})
$$
F_{C^\alpha}(v) = 
v + \sum^{n-1}_{k=1} \dfrac{1}{k!} \times (-1)^k \varphi^{-1(k)} (\varphi(t)) [\varphi(t)]^k.
$$
\item[4.] We set $u_1 = \varphi^{-1} (s_1 \cdots s_{N-1} \varphi (v))$,
$u_k = \varphi^{-1}\left( (1-s_{k-1}) \prod ^{N-1}_{j=k} s_j \cdot \varphi(v) \right)$
for $k=2,...,N-1$, and
$\ u_N = \varphi^{-1}((1-s_{N-1})\varphi(v))$.
\item[5.] The desired values are $x_k = F_k^{-1}(u_k)$, where we fixed $F_k=F$, $k=1,2,...,N$.
\end{enumerate}


\section{Optimal consumption algorithm}
\label{algorithm}

We present an algorithm to find the optimal consumption process. To do this, we consider a grid of $n$ states of the market at each time-tick $k$, denoting the state space as $\{\omega_i\}^n_{i=1}$. 

We assume that the financial market is free of arbitrage and that there exists a price kernel $\xi=\xi_k$ at each time $k$, such that $X_k$ is priced as $\mathbb{E}[\xi X_k]= \mathbb{E}_\Q [X_k]$, where $\Q$ is the risk neutral measure. Without loss of generality, we can consider $\xi(\omega_i)\leqslant \xi(\omega_{i+1})$, $1 \leqslant i \leqslant n-1 $.

The process of optimal consumption until horizon time-tick $N$ satisfies the following equation:
\begin{equation}
\stackbin[(X_1,X_2, ...,X_N)\sim C]{}{\textup{min}}  \mathbb{E} \Bigg[ \xi_N \Bigg( \sum_k X_k \Bigg) \Bigg] , 1 \leq k \leq N,
\end{equation}
where $X_k \sim F_k$ for fixed desired distributions $F_k$. We represent with $C$ the copula that describes the relation between the random variables $X_k$.\\

We describe the algorithm  in three steps:

\begin{enumerate}

\item Simulate the state-price process:

We generate $n$ simulated values of the state-price process $\xi(i) =\xi(\omega_i)$, $1 \leqslant i \leqslant n $ , at time $N$. 

\item Simulate the values of the consumption distribution:

We generate $n$ simulated values of $N$ correlated vectors $C_1, ..., C_N$ of the chosen copula.

From these vectors, we compute $X_1 = F_1^{-1}(C_1), ..., X_N = F_N^{-1}(C_N)$, which are $N$ correlated vectors with distributions $F_1,...,F_N$ respectively, and set  $Z = \sum_k X_k$. 

Then we have the following arrangement 

\begin{equation}
\begin{array}{cccc|c|c}
x_{11} & x_{21} & \ldots &  x_{N1} & \sum_k x_{k1} = z_1 & \xi(1) \\
\vdots & \vdots & \ddots &  \vdots & \vdots & \vdots \\
x_{1n} & x_{2n} & \ldots &  x_{Nn} & \sum_k x_{kn} = z_n & \xi(n)
\end{array}
\end{equation}

\item Solution:

We construct $Z^*$ from the elements of $Z$ ordered to be antimonotonic with the state price vector. It means that we arrange the observations $z_j$ increasingly to be antimonotonic with the last column.

\begin{equation}
\begin{array}{c|c}
z_1^*  & \xi(1)  \\
\vdots & \vdots \\
z_n^*  & \xi(n)  \\
\end{array}
\end{equation}

Let $p(j)$ be the permutation function $p:\{ 1,..,n\}\rightarrow\{1,...,n\}$ such that for every $j$, $z_{p(j)} = z^*_{i}$ for some $i$ in $\{1,...,n\}$. Therefore, the vector $X_k^*$, $1 \leq k \leq N$, that improves the original one is 
\begin{equation}
X_k^* = \left[
\begin{array}{c}
x_{t \sigma(1)} \\
\vdots \\
x_{t \sigma(n)}
\end{array}
\right].
\end{equation}

The variables $X_1^*,...,X_N^*$ have the same joint distribution and a lower or equal cost with respect to the original ones, and we achieve 
\begin{equation}
\stackbin[(X_1,X_2, ...,X_N)\sim C]{}{\textup{min}} \mathbb{E} \Bigg[ \xi_N \Bigg( \sum_k X_k \Bigg) \Bigg] \approx  \mathbb{E} \Bigg[ \xi_N \Bigg( \sum_k X_k^* \Bigg)  \Bigg] 
= \mathbb{E} \left[ \xi_N Z^* \right].
\end{equation}
\end{enumerate}


\section{Examples of implementation}
\label{examples}

We show an implementation of the intertemporal optimal consumption. Our first example is under the Black-Scholes model, where we can explicitly compute the state price process. The second one is under the Constant Elasticity of Variance (CEV) model, where we obtain the distribution of the price kernel via the inverse Laplace transform and reproduce its values through a simulation. In both examples, the number of periods is $N=10$.


For both models, we assume the investor chooses a consumption lognormal distribution $F$ with parameters $\mu_{log}$ and $\sigma_{log}$ (which are kept the same for every period for ease of computation). We model the dependency between them using a Clayton copula.


\subsection{Black-Scholes model}
$•$

Under the Black-Scholes model, the dynamics of $S_t$ are
\begin{equation}
dS_t = \mu S_t dt + \sigma S_t dW_t, \ S_0>0,
\label{S_t_dynamics}
\end{equation}
%
%
where $W_t,\ 0\leqslant t \leqslant T$, is a Brownian motion on a probability space $(\Omega,\mathbb{F},\mathbb{P})$, where $\mathbb{F}=\{\mathcal{F}_t\}_{0\leq t\leq 1},$ is the natural filtration for this Brownian motion.

Therefore, the dynamics of the discounted price process $\tilde{S}_t = e^{-rt} S_t$ are
\begin{equation}
d\tilde{S}_t = (\mu - r) \tilde{S}_t dt + \sigma \tilde{S}_t d W_t, \ \tilde{S}_0>0.
\label{S_t_disc_dynamics}
\end{equation}


To find the state price process or price kernel under the Black-Scholes model, we require to find the risk-neutral measure $\mathbb{Q}$ in virtue of Girsanov theorem, such that the discounted price process $\tilde{S}_t = e^{-rt} S_t$ is a $\mathbb{Q}$-martingale.

We let $\varphi _t,\ 0\leqslant t \leqslant T$, be an $\mathbb{F}$-adapted process such that $\tilde{S}_t$ is a martingale with dynamics 
$
d\tilde{S}_t = \tilde{\sigma} d W^{\Q}_t,
$
where 
$$ 
W^{\Q}_t = W_t + \int_0^t \varphi_s ds.
$$

To find $\varphi_t$ explicitly, we rewrite the dynamics of $\tilde{S}_t$ as 
$$
\begin{array}{ll}
d\tilde{S}_t &= (\mu - r) \tilde{S}_t dt + \sigma_t \tilde{S}_t \big(dW_t^{\Q} -\varphi_t dt\big)
\\\\
&= \left[ (\mu - r) - \sigma\varphi_t \right]\tilde{S}_t dt + \sigma_t \tilde{S}_t d W^{\Q}_s,
\end{array}
$$
set $(\mu - r) - \sigma\varphi_t = 0$,  and find that  $\varphi_t = \theta$, where $\theta =  \frac{\mu - r}{\sigma}$, the Sharpe ratio of the stock. 

Note that  $\varphi_t = \varphi$ satisfies the Novikov's condition since it is constant, and therefore, the stochastic exponential 
$$
Z_t = \exp \left\lbrace - \displaystyle\int_0^t \varphi _u dW_u - \dfrac{1}{2} \int_0^t \varphi^2 _u du \right\rbrace
$$
is a martingale under the probability measure $\mathbb{Q}$, given by 
$$
\mathbb{Q}(A) = \int_A Z(\omega) dP(\omega), \textup{ for all } A \in \mathcal{F},
$$ 
and the process $W^\mathbb{Q}_T, \ 0\leqslant t \leqslant T$, is a $\mathbb{Q}$-Brownian motion.

In consequence, the measures $\mathbb{P}$ and $\mathbb{Q}$ are related by the Radon-Nikodym derivative
\begin{equation}
\dfrac{d\mathbb{Q}}{d\mathbb{P}} = \exp \left( -\int_0^t \dfrac{\mu - r}{\sigma} dWs - \dfrac{1}{2} \int_0^t \left( \dfrac{\mu - r}{\sigma} \right)^2 dt \right),
\label{Radon-Nikodym}
\end{equation}
and we have that the state price process $\xi_t$ under the Black-Scholes model is
\begin{equation}
\xi_t := e^{-rt} \frac{d\mathbb{Q}}{d\mathbb{P}} = e^{-rt} e^{-\frac{1}{2} \left( \frac{\mu-r}{\sigma} \right)^2t} e^{-(\frac{\mu-r}{\sigma})W_t}.
\label{xi_t}
\end{equation}

We rewrite the state price process at time $T$ in terms of $S_T$
\begin{equation}
\xi_T = a \left( \dfrac{S_T}{S_0} \right)^{-\frac{\theta}{\sigma}},
\label{xi_T}
\end{equation}
where 
$$a = \exp \left( \dfrac{\theta}{\sigma} \left( \mu - \dfrac{\sigma^2}{2} \right) T - \left( r + \dfrac{\theta^2}{2} \right) T \right).$$ 


Under the Black-Scholes model, the unique state-price process at time $T$, is given in \eqref{xi_T}. Assuming that the stock expected return is bigger than the risk-free rate, \textit{i.e.}, $\mu>r$, we see from the previous expression that $\xi_T$ is anticomonotonic with ${S_T}$ since 
$\xi_T(\omega_i)$ is equal to
 $ \frac{a S_0}{{S_T}(\omega_i)^{\ \theta / \sigma}} 
$, where $a$ is a positive constant, and we use $\{\omega_i\}^n_{i=1}$ as a grid of the state space.

Having obtained the expression for the price kernel, we proceed with the first step of the algorithm outlined in Section \ref{algorithm}—that is, simulating the state-price process. Subsequently, we move on to step two, where we simulate the values of the consumption distribution while incorporating a dependency structure defined by a copula. Finally, in step three, we compute the cost of the strategy. 

We carry out these steps within a Black-Scholes model and 10 periods of consumption, employing the following market parameters: an expected rate of return of the stock ($\mu = 0.03$), a stock volatility ($\sigma = 0.3$), a risk-free rate ($r = 0.02$), and we consider a lognormal distribution with parameters $\mu_{log} = 100$ and a standard deviation of $40$ as our target distribution. Additionally, we employ a Clayton copula to establish the dependency structure, utilizing different values of $\alpha$.

We show in Figure~\ref{Fig3D} the resulting relationship between the strategy cost, the parameter alpha of the Clayton copula, and the standard deviation of the lognormal distribution, which is different from the parameter $\sigma_{log}$. We fixed the parameter $\mu_{log}=100$.

\begin{figure}[h!]
\begin{center}
\includegraphics[
trim = 0cm 0cm 0cm 0cm, clip, 
width=13cm]{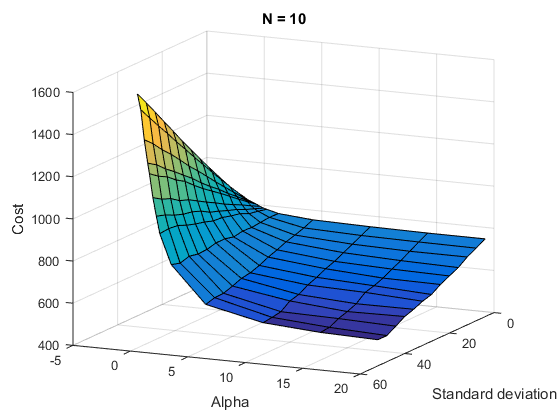}
\caption{Relation between the cost, the parameter alpha, and the standard deviation of the lognormal distribution (which is different from the parameter $\sigma_{log}$) for $N=10$ under Black-Scholes model and $\mu_{log}=100$.}\label{Fig3D}
\end{center}
\end{figure}

We notice that for negative values of alpha, the standard deviation of the chosen lognormal distribution is inversely proportional to the cost, and for positive values, the relationship is proportional. Then, the best choice to get a cost-efficient consumption is a positive alpha (\textit{e.g.}, $\alpha=5$ or $\alpha=20$).

Fixing the standard deviation to $40$, we see that the value of alpha that achieves the minimum cost is $20$, as we show in Figure~\ref{Fig2D}.

\begin{figure}[h!]
\begin{center}
\includegraphics[
trim = 0cm 0cm 0cm 0cm, clip, 
width=11cm]{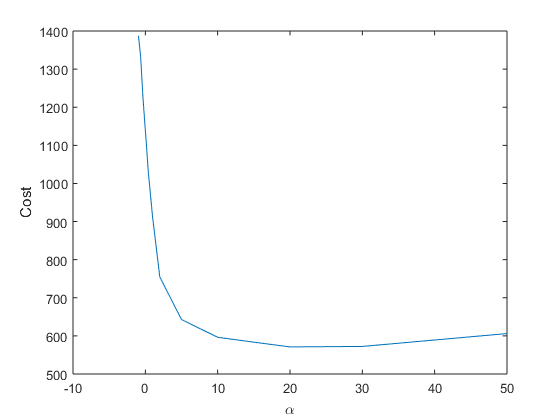}
\caption{Relation between the cost and alpha for a standard deviation of $40$, and $N=10$ under Black-Scholes model.}\label{Fig2D}
\end{center}
\end{figure}

Figure~\ref{FigCEF} displays what we refer to as the \textit{Cost-Efficient Frontier} for $\alpha=5$ and $\alpha=10$. The horizontal axis represents the standard deviation, acting as the input, while the vertical axis showcases the expected value of the lognormal distribution, which is the expected consumption per period. We have set $N=10$ periods and a fixed budget of $1000$.

\begin{figure}[h]
\begin{center}
\includegraphics[
trim = 0cm 0cm 0cm 0cm, clip, 
width=12cm]{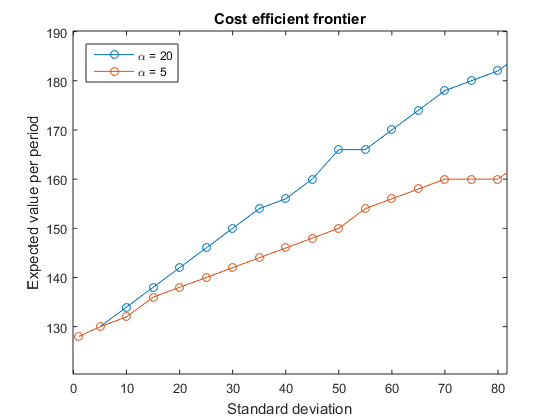}
\caption{Cost efficient frontier for two values of alpha and $N=10$ under the Black-Scholes model. The X-axis represent the standard deviation of the target distribution and the Y-axis the expected value per period of the target distribution.}\label{FigCEF}
\end{center}
\end{figure}

We see that we have bigger expected values for $\alpha=20$ than $\alpha=5$. The curve of the second case grows less fast. For bigger values of standard deviation, both curves show the classical concave shape of risk aversion behavior.

The cost-efficient frontier tells us that given the market conditions described above, if we assume a risk of $\$40$ for example, we can expect to win an average of $\$155$ over $10$ periods with a budget of $\$1000$ using a consumption correlation structure given by $\alpha=20$.\\


\noindent\textbf{Hedging strategy}

Within this section, we construct the hedging strategy to replicate the cost-efficient consumption process in a Black-Scholes market.


On the basis of the fact that $Z^*_T$ is anticomonotonic with $\xi_T$ and that it is the cheapest way to achieve the distribution $F$, we have that
\begin{equation}
Z^*_T = F^{-1} (1 - F_{\xi_T}(\xi_T)).
\end{equation}

Under the Black-Scholes model, the state price process is anticomonotonic with $S_T$ when $\theta>0$, or equivalently, $\mu>r$. Assuming this, we have
\begin{equation}
Z^*_T = F^{-1} (F_{S_T}(S_T)).
\end{equation} 

To find the hedging strategy for $Z^*_T$ consisting on holding $\delta_t$ units of the stock and $\psi_t$ units of the bond, we let $E_t$ denote the value of $Z^*_T$ at time $t$ under the measure $\Q$, i.e.,
$$E_t = B_t \mathbb{E}_\mathbb{Q} [B_T^{-1} Z^*_T \vert \mathcal{F}_t]
= e^{-r(T-t)} \mathbb{E}_\mathbb{Q} [Z^*_T \vert \mathcal{F}_t],$$
where $B_t=B_0 e^{rt}$ is the standard bond process, and we assume that $E_t$ is Malliavin differentiable and admits the Clark-Ocone representation formula:
\begin{equation}
E_t = \mathbb{E}_\Q [E_t] + \int_0^t \phi_s d W^{\Q}_s,
\label{ClarkOcone}
\end{equation}
where 
$\phi_t = \mathbb{E}_\Q [D_t E_t \vert \F_t]$.
We use the following expression of $S_t$ under $\Q$ 
$$ d {S_t} = r S_t + \sigma {S_t} dW^{\Q}_t$$ 
and substitute it in \eqref{ClarkOcone} to find that
\begin{equation}
E_t = E_0 + \int_0^t \frac{\phi_s}{\sigma S_s} dS_s
\ - \int_0^t \frac{r \phi_s}{\sigma} \ ds.
\label{ClarkOconeStock}
\end{equation}

Equation \eqref{ClarkOconeStock} is the integral form of the SDE
\begin{equation}
dE_t = \delta_t dS_t - \frac{\phi_t}{\sigma B_t} dB_t,
\end{equation}
where $\delta_t = \dfrac{\phi_t}{\sigma S_t}$ represents the units of the stock we need to hold. The units of the bond we need are given by
$$
\psi_t = e^{-r(T-t)} \mathbb{E}_\mathbb{Q} [Z^*_T \vert \mathcal{F}_t]
-\delta_t S_t,
$$
and the value of the position at time $t$, is $E_t = \delta_t S_t + \psi_t B_t = 0$.\\


\noindent\textit{Lognormal example}\\

To illustrate the application of these results with closed-form expressions, we show an example when $F$ is a lognormal distribution with parameters $\mu_\textup{log}$ and $\sigma_\textup{log}$. Then $F(x) = \Phi \left( \dfrac{\ln(x)-\mu_\textup{log}}{\sigma_\textup{log}} \right)$, where $\Phi$ is the c.d.f. of the standard normal distribution.

In this case, we find that
$$
\begin{array}{ll}
Z^*_T &= \exp \left\lbrace \mu_{\textup{log}} + \sigma_{\textup{log}} \Phi^{-1} \left( \Phi \left( \dfrac{\ln S_T - \mu}{\sigma} \right) \right)  \right\rbrace
\\\\
&= \exp \left\lbrace \ln \left( S_T^{\frac{\sigma_{\textup{log}}}{\sigma}} \right) + \mu_{\textup{log}} - \mu \dfrac{\sigma_{\textup{log}}}{\sigma} \right\rbrace
\\\\
&= b S_T^{\frac{\sigma_{\textup{log}}}{\sigma}},
\end{array}
$$
where $b = \exp \left\lbrace  \mu_\textup{log} - \mu \dfrac{\sigma_\textup{log}}{\sigma} \right\rbrace $. Hence, we obtain,

$$
\begin{array}{ll}
e^{r(T-t)} E_t &= \mathbb{E}_\Q \left[ b \left( S_0 \exp \left\lbrace -\dfrac{\sigma^2}{2} T + \sigma W_T^{\Q} \right\rbrace \right) ^{\frac{\sigma_{\textup{log}}}{\sigma}} \Bigg\vert \mathcal{F}_t \right]
\\\\
&= b S_0^{\frac{\sigma_{\textup{log}}}{\sigma}} \exp \left\lbrace -\dfrac{\sigma_{\textup{log}}\sigma}{2} \right\rbrace \mathbb{E}_\Q \left[ \exp \left\lbrace \sigma_{\textup{log}} W_T^\Q \right\rbrace \Big\vert \mathcal{F}_t \right]
\\\\
&= b S_0^{\frac{\sigma_{\textup{log}}}{\sigma}} \exp \left\lbrace -\dfrac{\sigma_{\textup{log}}\sigma}{2} + \sigma_{\textup{log}}^2 \dfrac{T-t}{2} + \sigma_{\textup{log}} W_t^\Q
\right\rbrace
\\\\
&= b \left( S_0 \exp \left\lbrace -\dfrac{\sigma^2}{2} + \dfrac{\sigma\sigma_{\textup{log}}(T-t)}{2} + \sigma W_t^\Q \right\rbrace \right)^{\frac{\sigma_{\textup{log}}}{\sigma}}
\\\\
&= b S_t^{\frac{\sigma_{\textup{log}}}{\sigma}} \exp \left\lbrace \sigma_{\textup{log}}^2 \dfrac{T-t}{2} \right\rbrace
\\\\
&= b S_t^{\sigma'} g_t,
\end{array}
$$
where $\sigma' = \frac{\sigma_{\textup{log}}}{\sigma}$ and $g_t = \exp \left\lbrace \sigma_{\textup{log}}^2 \dfrac{T-t}{2} \right\rbrace$.

Since $D_t S_t = \sigma S_t$, then
$D_t E_t = e^{-r(T-t)} b \sigma_{\textup{log}} S_t^{\sigma'} g_t$,
and finally, we have that the hedging positions are
\begin{equation}
\delta_t = e^{-r(T-t)} b \sigma' S_t^{\sigma'-1} g_t , \ \ 
\psi_t = e^{-r(T-t)} b (1-\sigma') S_t^{\sigma'} g_t.
\end{equation}

 
\subsection{CEV model}
$•$

We consider the case $S_t$ follows a Constant Elasticity of Variance (CEV) process, with dynamics
$$
dS_t = \mu S_t dt + \sigma S_t^{\beta + 1} dW_t ,
$$ 
where $\beta\in\mathbb{R}$. Then the discounted price process has the following dynamics
$$
d\tilde{S}_t = (\mu - r) \tilde{S}_t dt + \sigma \tilde{S}_t^{\beta + 1} dW_t \ .
$$

Observe that if $\beta=0$ and $\beta = -1$, the CEV model reduces to the Black-Sholes model and the Bachelier model, respectively.

According to \cite{mijatovic2012martingale}, we can use the density process 
\begin{equation}
Y_t = \exp \left\lbrace - \dfrac{\mu-r}{\sigma} \displaystyle\int_0^t \tilde{S}_s^{-\beta} dW_s 
- \dfrac{1}{2} \left( \dfrac{\mu-r}{\sigma} \right)^2 \int_0^t \tilde{S}_s^{-2\beta} ds \right\rbrace , \ t\in[0,T].
\label{density_CEV}
\end{equation}
to change the measure to the risk neutral measure $\mathbb{Q}$. Then $\tilde{S}_t$ is a martingale under $\mathbb{Q}$. In particular, $\tilde{S}_t$ is a CEV process that satisfies the driftless equation $d\tilde{S}_t = \sigma\tilde{S}^{\beta +1} dW^{\mathbb{Q}}$. Moreover, if $\beta < 0$, $Y$ is a uniformly integrable martingale, and if $\beta \geqslant 0$, $Y$ is a martingale that is  not uniformly integrable (see \cite{mijatovic2012martingale}).

To find the distribution of $Y_t$, we take the logarithmic version of $Y_t$ and write 
$$
\log Y_t = - \int_0^t \varphi_s dW_s - \dfrac{1}{2} \int_0^t \varphi^2 ds ,
$$
where $\varphi_s = \dfrac{\mu - r}{\sigma} \tilde{S}_s^{-\beta}$, and consider the moment generating function of $\hat{Y}_t = \log Y_t$
$$
\begin{array}{lll}
\mathbb{E}[e^{\lambda \hat{Y}_t}] &= \mathbb{E}\exp \left\lbrace \lambda \left( - \displaystyle\int_0^t \varphi_s dW_s - \dfrac{1}{2} \int_0^t \varphi_s^2 ds \right) \right\rbrace
\\\\ 
&= \mathbb{E} \exp \left\lbrace - \displaystyle\int_0^t \lambda \varphi_s dW_s - \dfrac{1}{2} \int_0^t (\lambda \varphi_s)^2 ds \right\rbrace \times
\\\\ 
& 
\ \ \ \exp \left\lbrace \dfrac{\lambda^2 - \lambda}{2} \displaystyle\int_0^t \left( \dfrac{\mu - r}{\sigma} \right)^2 \tilde{S}_s^{-2\beta} ds \right\rbrace.
\end{array}
$$ 

We define the measure $\mathbb{Q}_\lambda$ such that
$$
W_t^{\mathbb{Q}_\lambda} = W_t + \int_0^t \lambda \varphi_s ds.
$$
Observe that $\mathbb{Q} = \mathbb{Q}_1$. Under this measure, the stochastic exponential 
$$
\exp \left\lbrace - \displaystyle\int_0^t \lambda \varphi_s dW_s - \dfrac{1}{2} \int_0^t (\lambda \varphi_s)^2 ds \right\rbrace
$$
is a martingale for $\beta \in [0,1]$ in virtue of the theorem 4.1 in \cite{klebaner2014stochastic}. Here, the authors develop an extension of Bene\v s method.

Under $\mathbb{Q}_\lambda$, the discounted price process has the dynamics
$$
d \tilde{S}_t = (1 - \lambda) (\mu - r) \tilde{S}_t dt + \sigma \tilde{S}_t^{\beta+1} dW_t^{\mathbb{Q}_\lambda} ,
$$
and quadratic variation
$$
\langle \tilde{S}_t \rangle = \int_0^t \sigma^2 \tilde{S}_s^{2 (\beta+1) } ds\ .
$$

Subsequently, we can rewrite

\begin{equation}
\begin{array}{ll}
\mathbb{E} \left[ e^{\lambda \hat{Y}_t} \right] 
& = \mathbb{E}_{\mathbb{Q}_\lambda} \exp \left\lbrace \dfrac{\lambda^2 - \lambda}{2} \displaystyle\int_0^t \dfrac{(\mu - r)^2}{\sigma^4} \  \sigma^2 \tilde{S}_s^{-2\beta} ds \right\rbrace
\\\\
& = \mathbb{E}_{\mathbb{Q}_\lambda} \exp \left\lbrace \dfrac{\lambda^2 - \lambda}{2}\ \dfrac{(\mu - r)^2}{\sigma^4} \langle \log \tilde{S}' \rangle_t \right\rbrace,
\end{array}
\label{MGF_LT}
\end{equation}
where $\langle \log \tilde{S}' \rangle_t$ is the quadratic variation of $ \log\tilde{S}'_t$, where $\tilde{S}'_t$ is a CEV process with parameter $\beta' = -\beta$ and any drift. Let us take, for convenience, $\tilde{S}'_t$ with no drift.

We let $V_t = V(\tilde{S}'_t)$ be the instantaneous variance rate of $\tilde{S}'_t$, defined by $V(\tilde{S}'_t) = \sigma^2 (\tilde{S}'_t)^{2\beta'}$. We express $V_t$ as a quadratic drift 3/2 process 
\begin{equation}
dV_t = (p V_t + q V_t^2) dt + \epsilon V_t^{\frac{3}{2}} dW_t^{\mathbb{Q}_\lambda}, \label{variation}
\end{equation}
where $p = 0, \ q = \beta(2 \beta + 1)$, and $\epsilon = -2 \beta$ to apply the Theorem 3 of \cite{carr2007new} to \eqref{MGF_LT}, and find that the Laplace transform of $\langle \log \tilde{S}' \rangle_t$ is given by
\begin{equation}
\mathbb{E}^{\mathbb{Q}_\lambda} \exp \left\lbrace -s \left( \langle \log \tilde{S}' \rangle_T - \langle \log \tilde{S}' \rangle_t \right)
\left|\right. V_t = v \right\rbrace 
= \dfrac{\Gamma (\gamma_s - \alpha_s)}{\Gamma (\gamma_s)} \left( \dfrac{2}{\epsilon^2 h_t} \right)^{\alpha_s} M \left( \alpha_s : \gamma_s : - \dfrac{2}{\epsilon^2 h_t} \right),
\label{LT}
\end{equation}
where 

$$
\begin{array}{ll}
h_t = \displaystyle\int_t^T dt' V_0 = (T-t) \sigma^2 S_0^{-2 \beta},
\\\\
\alpha_s = -\left( \dfrac{1}{2} - \dfrac{q}{\epsilon^2} \right) + \sqrt{\left( \dfrac{1}{2} - \dfrac{q}{\epsilon^2} \right) + 2 \dfrac{s}{\epsilon^2}} 
= \dfrac{\sqrt{1 + 8 s}}{4 \beta} 
+ \dfrac{1}{4 \beta},
\\\\
\gamma_s = 2 \left( \alpha_s + 1 - \dfrac{q}{\epsilon^2	} \right) 
= 1 + 2 \alpha_s - \dfrac{1}{2(\beta+1)}
= \dfrac{\sqrt{1 + 8 s}}{2 \beta}  + 1,
\\\\
M(\alpha_s,\ \gamma_s,\ z) = \displaystyle\sum^\infty_{n=0} \dfrac{(\alpha_s)_n}{(\gamma_s)_n}\ \dfrac{z^n}{n!}.
\end{array}
$$
and $(x)_n $ denotes the Pochhammer's symbol: $ (x)_n = \Gamma (x+n) / \Gamma (x)$.

Our objective is to calculate the cumulative distribution function of the state price process to facilitate the consumption optimization process. To achieve this, one can either compute the inverse Laplace transform of the aforementioned expression or opt for Monte Carlo simulations of the paths of $\tilde{S}_t$. We have chosen to work with the latter approach due to its ease of programming.

We set $r_t = S_t^{2-\beta^*}$ with $\beta^*=2(\beta+1)$ and apply It\^o's formula to find that $r_t$ satisfies the SDE of the following square-root diffusion
%
%
\begin{equation}
dr_t = a (b-r_t)dt + \sigma' \sqrt{r_t} dW_t,
\end{equation}
where $a = -(\mu-r)(2-\beta^*)$, 
$b = (\sigma^2/2)(2-\beta^*)(1-\beta^*)/a$ 
and $\sigma' = \sigma (2-\beta^*)$.

We simulate $\tilde{S}_t$ for $1<\beta^*<2$ using the following algorithm presented in \cite{glasserman2003monte} which simulates the square root diffusion
$dr_t = \alpha (b-r_t)dt + \sigma' \sqrt{r_t} dW_t$ on time grid $0 = t_0<t_1<\cdots < t_n$:

\begin{center}
$
\begin{array}{lll}
\textup{Set}
d = 4b\alpha \big / \sigma'^2 
\\\\
\textup{Case 1: } d>1
\\
\textup{for } i = 0,...,n-1
\\
\ \ \ \ c \leftarrow \sigma'^2 (1-e^{-\alpha(t_{i+1}-t_i)}) \Big / (4\alpha) 
\\
\ \ \ \ \lambda \leftarrow r_{t_i} (e^{-\alpha(t_{i+1}-t_i)}) \Big / c
\\
\ \ \ \ \textup{generate } Z \sim N(0,1)
\\
\ \ \ \ \textup{generate } X \sim \chi^2_{d-1}
\\
\ \ \ \ r_{t_{i+1}}\leftarrow c\left[\left( Z + \sqrt{\lambda}\right)^2 + X \right]
\\
\textup{end}\\\\
\end{array}
$

$
\begin{array}{lll}
\textup{Case 2: } d\leqslant 1
\\
\textup{for } i = 0,...,n-1
\\
\ \ \ \ c \leftarrow \sigma'^2 (1-e^{-\alpha(t_{i+1}-t_i)}) \Big / (4\alpha) 
\\
\ \ \ \ \lambda \leftarrow r_{t_i} (e^{-\alpha(t_{i+1}-t_i)}) \Big / c
\\
\ \ \ \ \textup{generate } N \sim \textup{Poisson} (\lambda /2)
\\
\ \ \ \ \textup{generate } X \sim \chi^2_{d+2N}
\\
\ \ \ \ r_{t_{i+1}}\leftarrow cX
\\
\textup{end}
\end{array}
$
\end{center}

Finally, we invert the transformation from $\tilde{S}_t$ to $r_t$ to obtain $\tilde{S}_t = r_t^{1/(2-\beta^*)}$. Utilizing the expression \eqref{density_CEV}, we compute $\xi_T$, and subsequently, we implement the algorithm described in Section \ref{algorithm} with a dependence structure defined by a Clayton copula.

We present an example using the market parameters $\beta = -1/4$, $\mu = 0.03$, \ $r = 0.02, \ \sigma = 0.3$ with $N=10$ periods cf consumption. We fixed the parameter $\mu_{log}=100$.
In Figure~\ref{Fig3D_CEV}, we depict the resulting relationship between the strategy cost, the parameter alpha of the Clayton copula, and the standard deviation of the lognormal distribution, which is different from the parameter $\sigma_{log}$. 

\begin{figure}[h!]
\begin{center}
\includegraphics[
trim = 0cm 0cm 0cm 0cm, clip, 
width=13cm]{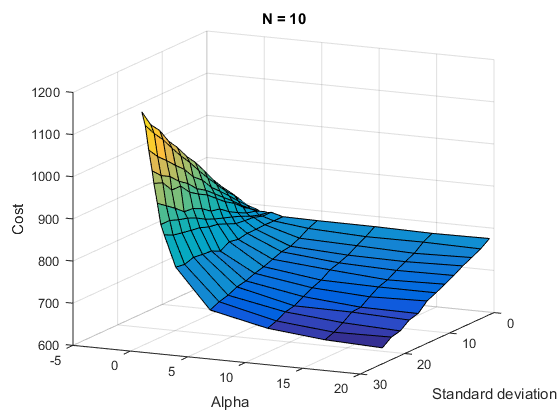}
\caption{Relation between the cost, the parameter alpha, and the standard deviation of the lognormal distribution for $\mu_{log} = 100$ and $N=10$ under CEV model.}\label{Fig3D_CEV}
\end{center}
\end{figure}

It is noteworthy that for negative values of alpha, the standard deviation of the chosen lognormal distribution is inversely proportional to the cost, whereas for positive values, the relationship is proportional, as in the Black-Scholes model. When comparing it with the Black-Scholes model, we observe that under the CEV model with $\beta = -1/4$, the range of the cost is smaller, but the overall shape of the plot remains the same.

We choose a standard deviation equal to $40$ as in Black-Scholes model to see that the value of alpha that achieves the minimum cost is $20$, as we see in Figure~\ref{Fig2D_CEV}.

\begin{figure}[h!]
\begin{center}
\includegraphics[
trim = 0cm 0cm 0cm 0cm, clip, 
width=11cm]{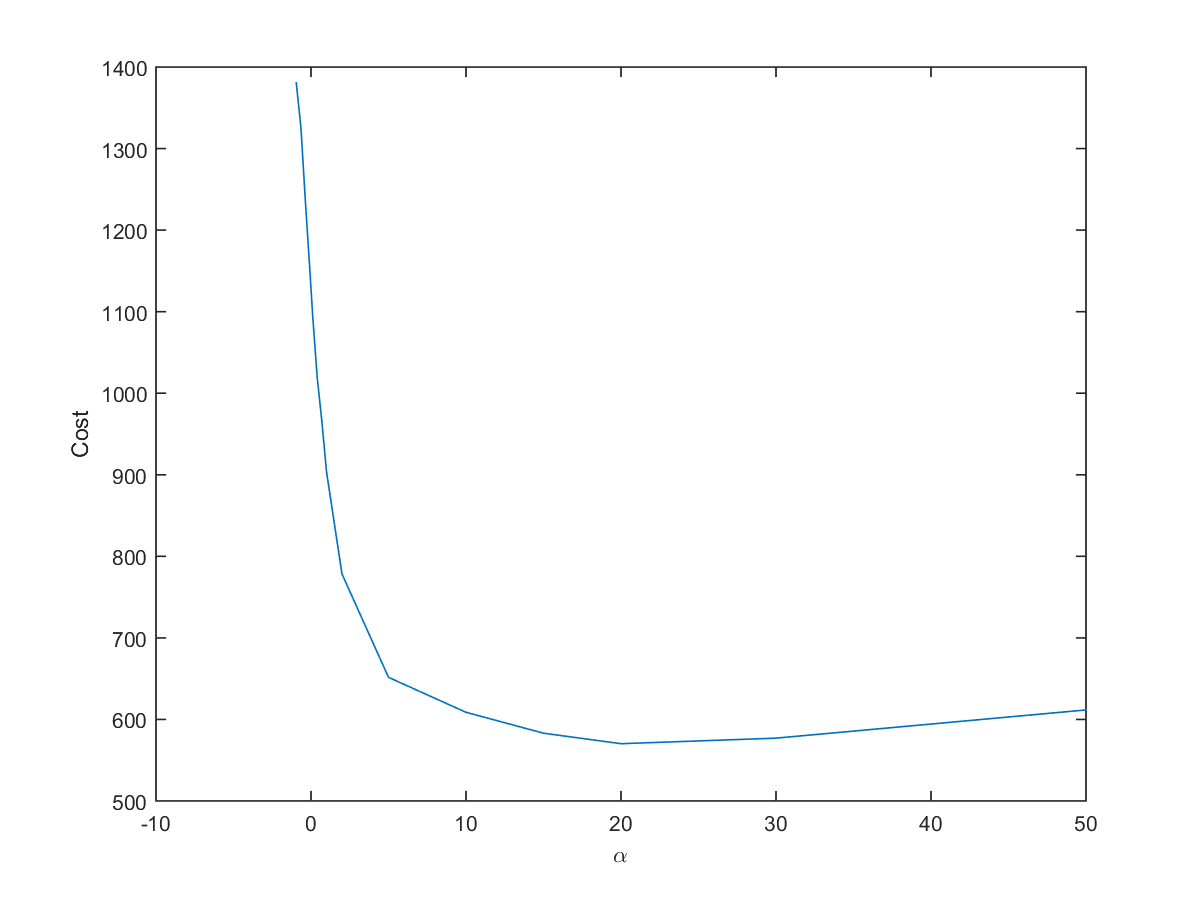}
\caption{Relation between the cost and alpha for a standard deviation of $40$, and $N=10$ under CEV model.}\label{Fig2D_CEV}
\end{center}
\end{figure}

In Figure~\ref{FigCEF_CEV}, we show the \textit{Cost-efficient frontier} for $\alpha=5$ and $\alpha=20$ under the CEV model. We plot the standard deviation as an input and the expected value of the lognormal distribution as an output, which is the expected consumption per period. We have set $N=10$ periods and a fixed budget of 1000.

\begin{figure}[h]
\begin{center}
\includegraphics[
trim = 0cm 0cm 0cm 0cm, clip, 
width=12cm]{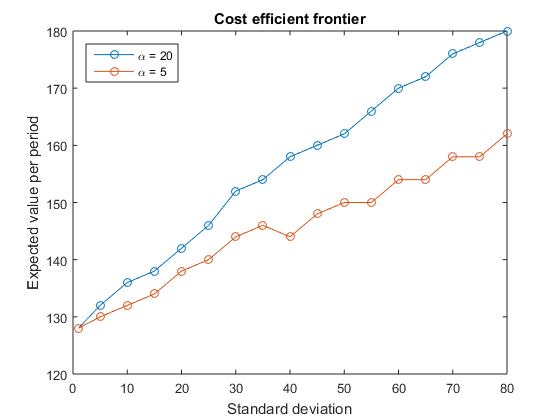}
\caption{Cost-efficient frontier for two values of alpha and $N=10$ under CEV model. The X-axis represent the standard deviation of the target distribution and the Y-axis the expected value per period of the target distribution.}\label{FigCEF_CEV}
\end{center}
\end{figure}

We observe that the expected value is higher for $\alpha=20$ compared to $\alpha=5$. This is related with the fact that the curve in the second case exhibits a slower growth rate. 

What the cost-efficient frontier tells us is that, for example, if we assume a risk of $\$40$, we can expect to win an average of $\$160$ over $10$ periods with a budget of $\$1000$ using a consumption correlation structure given by $\alpha=20$.

\section{Conclusions}

In this paper, we have explored intertemporal cost-efficient consumption, which represents an extension of the Distribution Builder approach to portfolio selection. We have addressed the problem of achieving desired consumption distributions while minimizing costs across multiple time periods.

The incorporation of a dependency structures using copulas is central to our method, these provide an elegant means of representing relationships between multiple variables. In our work, we opt for the Clayton copula due to its simplicity and ability to be modeled with a single parameter.

We have developed an algorithm designed to determine the cost-efficient consumption and offer the methodology to attain the optimal strategy within the Black-Scholes market and the Constant Elasticity of Variance (CEV) market. 

Our findings reveal that in both markets, positive correlated consumption random variables lead to cost-efficient strategies. We have also introduced a cost-efficient frontier for both cases, offering the expected return of each period according to the risk the agent is willing to face.

\textbf{Acknowledgements}\\
This project has received funding from the European Union’s Horizon
2020 research and innovation programme under the Marie
Skłodowska-Curie grant agreement No. 777822.\\


\bibliographystyle{alpha}
\bibliography{CEC_bib}

\end{document}